\documentclass[preprint,11pt]{elsarticle}
\usepackage{amsfonts}
\usepackage{verbatim}
\usepackage[fleqn]{amsmath}
\usepackage{amssymb}
\usepackage{latexsym}
\usepackage{graphicx}
\usepackage{bm}
\usepackage[usenames]{color}
\usepackage{booktabs}
\usepackage{lineno}
\usepackage{amsthm}

\input amssym.def
\newsymbol\rtimes 226F
\newfont{\nset}{msbm10}

\newtheorem{theo}{Theorem}[section]
\newtheorem{theorem}[theo]{Theorem}

\newtheorem{lemma}[theo]{Lemma}
\newtheorem{definition}[theo]{Definition}

\journal{Theoretical Computer Science}

\begin{document}

\begin{frontmatter}

\title{The number and degree distribution of spanning trees in the Tower of Hanoi graph}

\author{Zhongzhi Zhang}

\ead{zhangzz@fudan.edu.cn}

\author{Shunqi Wu, Mingyun Li}
\address{School of Computer Science, Fudan University, Shanghai 200433, China}

\address{Shanghai Key Laboratory of Intelligent Information Processing, Fudan University, Shanghai 200433, China}

\author{Francesc Comellas}
\address{Mathematics Department, Universitat Polit\`ecnica de Catalunya, Barcelona 08034, Catalonia, Spain}
\ead{francesc.comellas@upc.edu}


\begin{abstract}
The number of spanning trees of a graph is an important invariant related to topological and dynamic properties of the graph, such as its reliability, communication aspects, synchronization, and so on. {However, the practical
enumeration of spanning trees and the study of their properties remain a
challenge, particularly for large networks.} In this paper, we study the number and degree distribution of the spanning trees in the Hanoi graph. We first establish recursion relations between the number of spanning trees and other spanning subgraphs of the Hanoi graph, from which we find an exact analytical expression for the number of spanning trees of the $n$-disc Hanoi graph. This result allows the calculation of the spanning tree entropy which is then compared with those for other graphs with the same average degree. Then, we introduce a vertex labeling which allows to find, for each vertex of the graph, its degree distribution among all possible spanning trees.
\end{abstract}

\begin{keyword}
Spanning trees \sep Tower of Hanoi graph \sep Degree distribution \sep Fractal geometry
\end{keyword}


\end{frontmatter}

\section{Introduction }

The problem of finding the number of spanning trees of a finite
graph is a relevant and long standing question. 
It has been considered in different areas of mathematics~\cite{OxYa11}, 
physics~\cite{LiWuZhCh11},  and computer science~\cite{NiPaPa14}, 
since its introduction by  Kirchhoff in 1847~\cite{Ki1847}. 
This graph invariant is a  parameter that
characterizes the reliability of a network~\cite{Bo86,Co87,PeBoSu98}
and is related to its optimal synchronization~\cite{TaMo06} and
the study of random walks~\cite{Al90}. It is also of interest in
theoretical chemistry, see for example~\cite{BrMaPoRo96}. The number
of spanning trees of a graph can be computed, as shown in many basic
texts on graph theory~\cite{GoRo01}, from Kirchhoff's matrix-tree theorem~\cite{ChKl78} and
it is given by the product of all nonzero eigenvalues of the
Laplacian matrix of the graph. 
Although this result
can be applied to any graph,  the calculation of the number of
spanning trees from the matrix theorem is analytically and
computationally demanding, in particular for large networks. 
Not surprisingly,  recent work has been devoted to finding alternative
methods to produce closed-form expressions for the number of
spanning trees for particular graphs such as grid graphs~\cite{NiPa04}, 
lattices~\cite{Wu77,ShWu00,ZhLiWuZo11,ZhWu15}, 
the small-world Farey graph~\cite{ZhCo11,ZhWuLi12,YiZhLiCh15}, 
the Sierpi\'nski gasket~\cite{TeWa06,ChChYa07}, 
self-similar lattices~\cite{TeWa11,TeWa11JSP}, etc.

Most of the previous work focused on counting spanning trees on 
various graphs~\cite{OxYa11}. 
However, the number of spanning trees  is an integrated,
coarse characteristic of a graph. 
Once the number of spanning trees is determined, the next step 
is to explore and understand the geometrical structure of spanning trees. 
In this context, it is of great interest to compute the probability distribution of different 
coordination numbers at a given vertex among all the spanning trees~\cite{Al91}, 
which encodes useful information about the role the vertex plays in the whole network. 
Due to the computational complexity of the calculation, this geometrical feature of 
spanning trees has been studied only for very few graphs, such as the 
$\mathbb{Z}^d$ lattice~\cite{BuPe93}, the square lattice~\cite{MaDhMa92}, 
and the Sierpi\'nski graph~\cite{ChCh10}. 
It is non-trivial to study this geometrical structure for other graphs.

In this paper, we study the number and structure of spanning trees
of the Hanoi graph. This graph, which is also known as the Tower of
Hanoi graph~\cite{HiKlMiPeSt13}, comes from the well known Tower of Hanoi puzzle,
as the graph is associated to the allowed moves in this puzzle. 
There exist an abundant literature on the
properties of the Hanoi graph, which includes the study of shortest
paths, average distance, planarity, Hamiltonian walks, group of
symmetries, average eccentricity, to name a few, see~\cite{HiKlMiPeSt13} and
references therein. 
In~\cite{TeWa11JSP}, Teufl and Wagner obtained the number of spanning trees 
of different self-similar lattices, including the Hanoi graph. 
{Here, based on the self-similarity of the Hanoi graph, we enumerate 
its spanning trees and compute for each vertex of the graph its degree distribution
among all spanning trees.}

\section{The Hanoi graph}

The Hanoi graph is derived from the Tower of Hanoi puzzle with $n$ discs~\cite{HiKlMiPeSt13}. We can  consider each legal distribution of the $n$ discs on the three peg, a state, as a vertex of the Hanoi graph, and an edge is defined if one state can be transformed into another by moving one disc. If we label the three pegs 0, 1 and 2,  any legal distribution of the $n$ discs can be written as the vector/sequence
$\alpha_1\alpha_2\ldots\alpha_n$ where $\alpha_i$ ($1\leq i \leq n$) gives
the location of the $(n+1-i)$th largest disc. We will denote as $H_n$  the Hanoi
graph of $n$ discs. Fig.~\ref{han} shows $H_1$, $H_2$ and $H_3$.

\begin{figure}[htb]
\begin{center}
\includegraphics[width=0.5\textwidth,angle=90,trim=40 60 100 40]{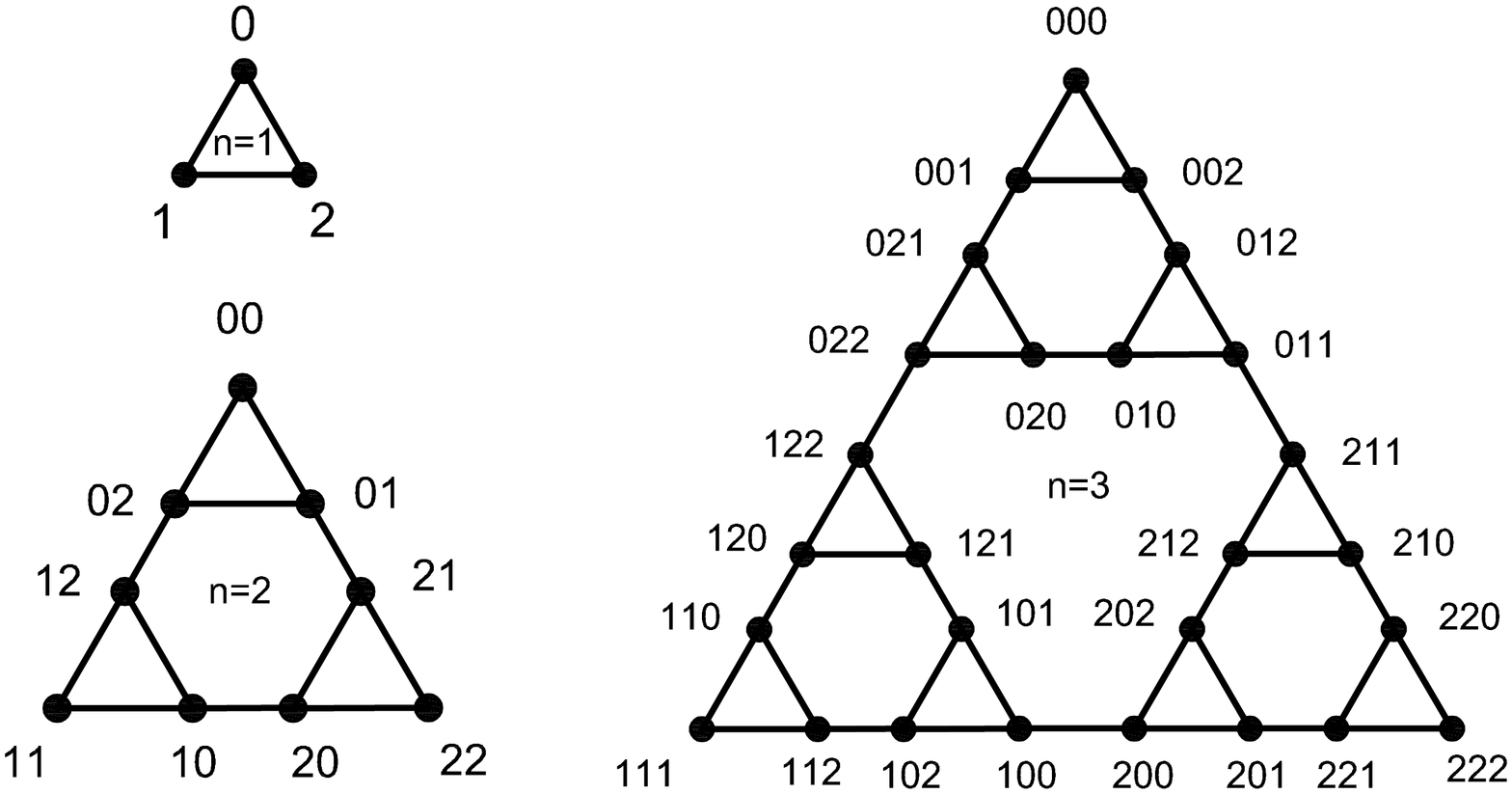}
\end{center}
\caption{Hanoi graphs $H_1$, $H_2$ and $H_3$.}\label{han}
\end{figure}

Note that $H_{n+1} (n\ge 1)$ can be obtained from three copies of
$H_{n}$ joined by three edges, each one connecting a pair of vertices
from two different replicas of $H_{n}$, as shown in Fig.~\ref{str}.
From the construction rule, we find that the number of vertices or
order of $H_n$ is $3^n$ while the number of edges is
$\frac{3}{2}(3^n-1)$.

\begin{figure}[htb]
\begin{center}
\includegraphics[width=0.6\textwidth]{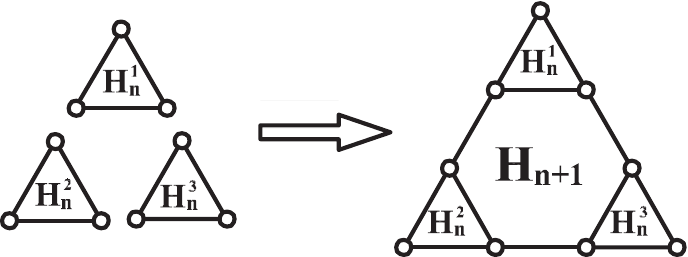}
\end{center}
\caption{Construction rules for the Hanoi graph. $H_{n+1}$ is
obtained by connecting three graphs $H_{n}$ labeled here by
$H^1_{n}$, $H^2_{n}$ and $H^3_{n}$.}\label{str}
\end{figure}

In the next section will make use of  this recursive construction to
find the number of spanning trees of $H_n$ at any iteration step
$n$.

\section{The number of spanning trees in $H_n$}

If we denote by $V_n$ and $E_n$ the number of vertices and edges of
$H_n$, then a spanning subgraph of $H_n$ is a graph with the same
vertex set as $H_n$ and a number of edges $E_n^\prime$ such that
$E_n^\prime < E_n$. A spanning tree of $H_n$ is a spanning subgraph
that is a tree and thus $E_n^\prime = V_n-1$.

In this section we calculate the number of spanning trees of the
Hanoi graph $H_n$. We adapt the decimation method described
in~\cite{Dh77,DhDh97,KnVa86}, which has also been successfully used to
find the number of spanning trees of the Sierpi\'nski
gasket~\cite{ChChYa07}, the Apollonian network~\cite{ZhWuCo14}, and
some fractal lattices~\cite{ZhLiWuZo11}. 
This decimation method is in fact the standard renormalization group a
pproach~\cite{Wi1975} in statistical physics,  which applies to many enumeration 
problems on self-similar graphs~\cite{TeWa07}.  
We make use of the
particular structure of the Hanoi graph to obtain a set of recursive
equations for the number of spanning trees and spanning subgraphs,
which then can be solved by induction.

{Let ${\rm S}_n$ denote the set of spanning trees of $H_n$. Let ${\rm
P}_n$ (${\rm R}_n$, ${\rm T}_n $) denote the set of spanning subgraphs of $H_n$, 
each of which consists of two trees with the outmost vertex $22\ldots 2$ ($00\ldots 0$, $11\ldots 1$) 
belonging to one tree while the other two outmost vertices being in the second tree. 
And let ${\rm L}_n$ denote the set of spanning subgraphs of $H_n$, 
each of which contains three trees with every outmost vertex in a different tree. 
These five types of spanning subgraphs are illustrated schematically in Fig.~\ref{sub}, 
where we use only the three outmost vertices to represent the graph because the edges 
joining the subgraphs to which they belong provide all the information needed to obtain the 
Hanoi graph at the next iteration. 
Let $s_n$, $p_n$, $r_n$, $t_n$, and $l_n$ denote the cardinality of sets 
${\rm S}_n$, ${\rm P}_n$, ${\rm R}_n$, ${\rm T}_n$, and ${\rm L}_n$, respectively.}
\begin{figure}[ht]
\begin{center}
\includegraphics[width=0.25\textwidth,trim=170 480 180 550,angle=90]{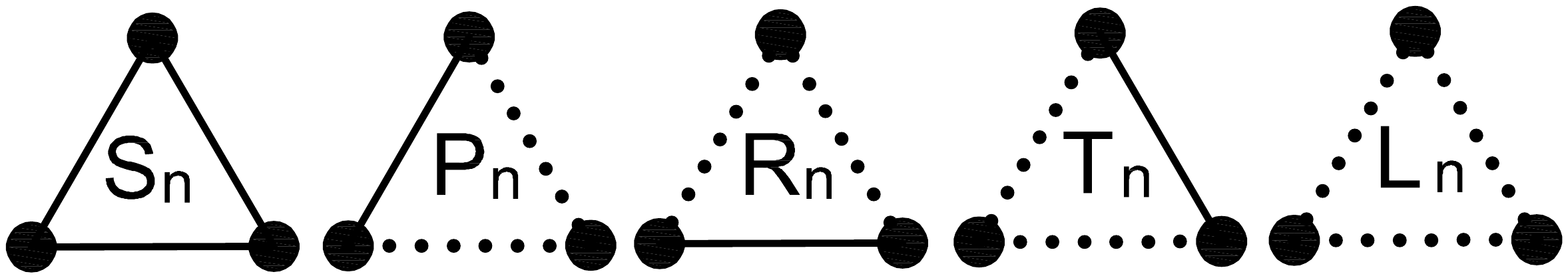}
\end{center}
\caption{Illustration for the five types  of spanning subgraphs
derived from $H_n$. Two outmost vertices joined by a solid line are in one tree
while two outmost vertices belong to different trees if they are connected by a
dashed line.}\label{sub}
\end{figure}

\begin{lemma}
\label{th:sub} The five classes of subgraphs ${\rm S}_n$, ${\rm P}_n$, ${\rm
R}_n$, ${\rm T}_n $ and ${\rm L}_n$  form a complete set because
each one can be constructed iteratively from the classes of subgraphs ${\rm
S}_{n-1}$, ${\rm P}_{n-1}$, ${\rm R}_{n-1}$, ${\rm T}_{n-1}$ and
${\rm L}_{n-1}$.
\end{lemma}

We do not prove this Lemma here, since we will enumerate each case.
However the proof follows from  the fact that $H_{n}$  can be
constructed from three $H_{n-1}$ by joining their outmost vertices and each of
the five subgraphs are associated with different ways to produce the
spanning trees.

Next we will establish a recursive relationship among the five
parameters $s_n$, $p_n$, $r_n$,  $t_n$ and $l_n$. We notice that the
equation $p_n =  r_n  = t_n$ holds as a result of symmetry, thus, in
some places of the following text, we will use $p_n$ instead of $r_n$ and $t_n$.

\begin{lemma}
\label{th:fn} For  the Hanoi graph $H_n$ with $n \geq 1$,
\begin{equation}\label{eq:sna}
s_{n+1}  =  3s^3_{n}+6s^2_n p_n\,,
\end{equation}
\begin{equation}\label{eq:snb}
p_{n+1} =  s^3_n+7s^2_np_n+7s_np^2_n+s^2_nl_n \,,
\end{equation}
\begin{equation} \label{eq:snc}
l_{n+1}  =  s^3_n+12s^2_n p_n+3s^2_n l_n+36s_n p^2_n+12s_n p_n l_n+14 l^3_n\,.
\end{equation}
\end{lemma}
\begin{proof}
This lemma  can be proved graphically. Fig.~\ref{fn} shows a graphical representation 
of Eq.~\eqref{eq:sna}. Fig.~\ref{pn} provides a case enumeration 
for $p_{n+1}$. Fig.~\ref{ln} and Fig.~\ref{lnplus} give the enumeration detail 
of all configurations that contribute to $l_{n+1}$.
\end{proof}
\begin{figure}[ht]
\begin{center}
\includegraphics[width=0.30\textwidth,trim=120 380 150 350,angle=90]{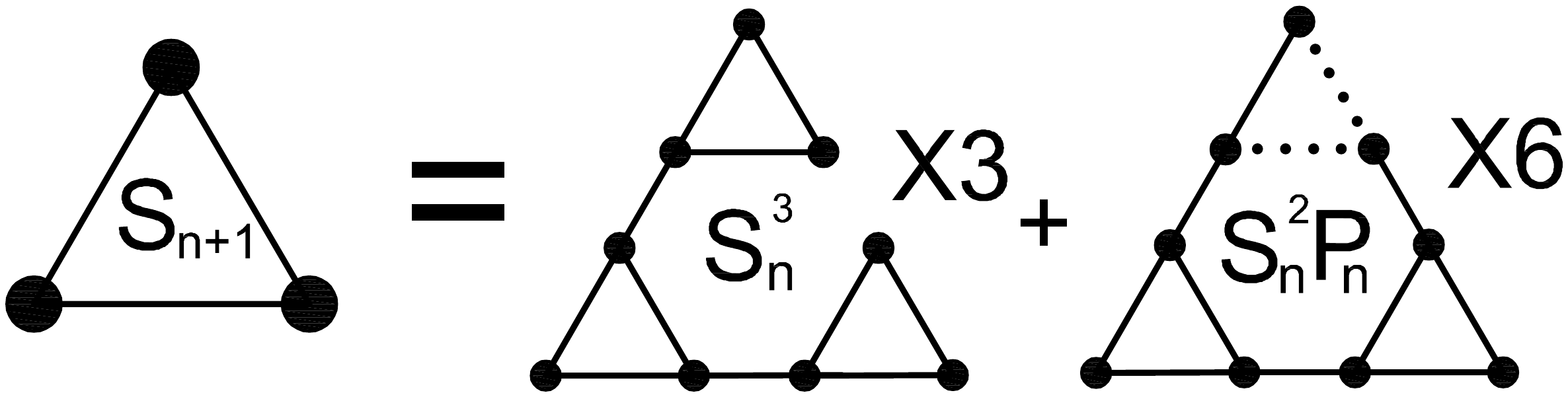}
\end{center}
\caption{Illustration of the configurations needed to find $s_{n+1}$.} \label{fn}
\end{figure}

\begin{figure}[htb]
\begin{center}
\includegraphics[width=0.7\textwidth,trim=30 130 30 80,angle=90]{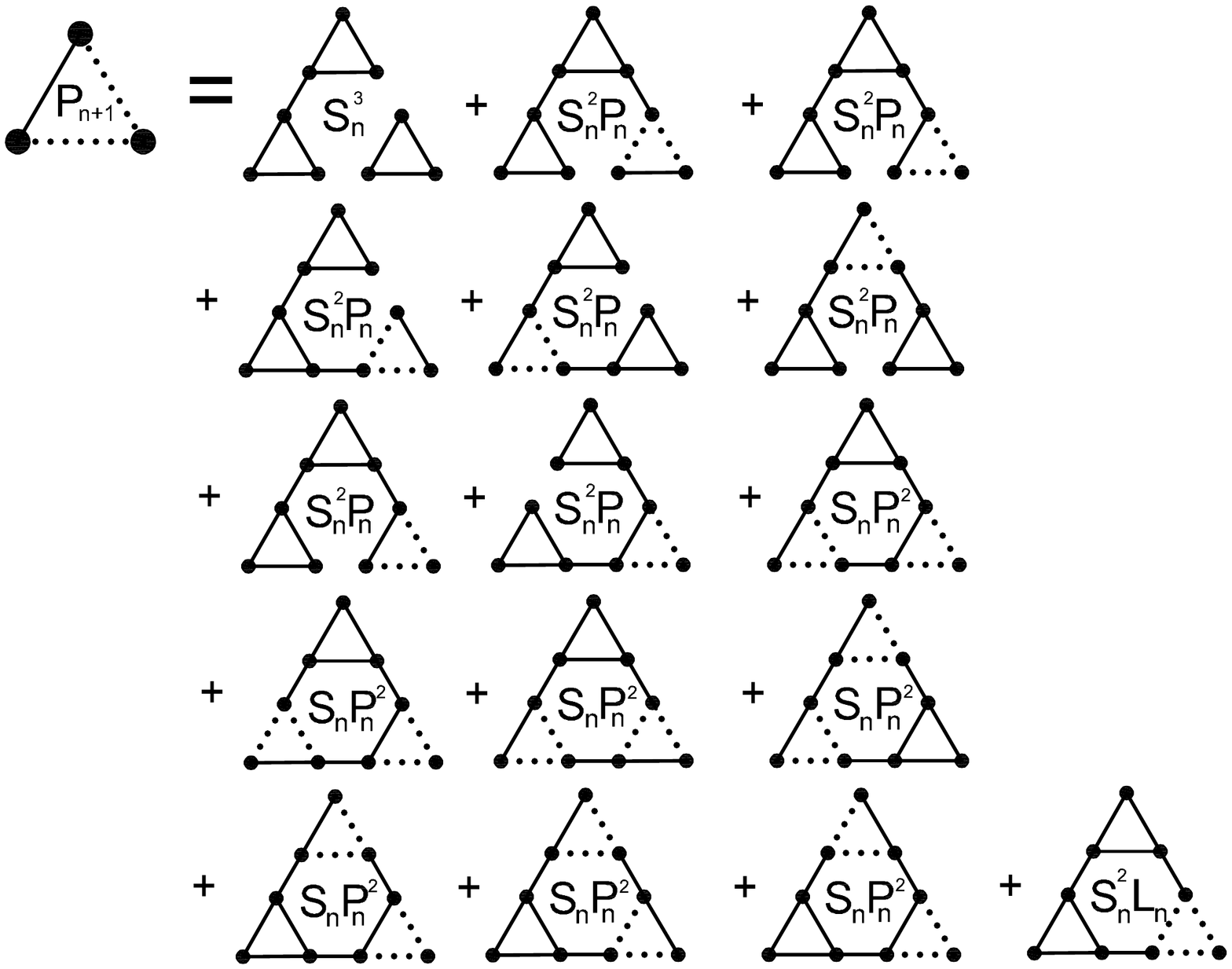}
\end{center}
\caption{{Illustration of the configurations needed to find $p_{n+1}$.}}
\label{pn}
\end{figure}

\begin{figure}[htb]
\begin{center}
\includegraphics[width=0.6\textwidth,angle=90,trim=30 0 0 0]{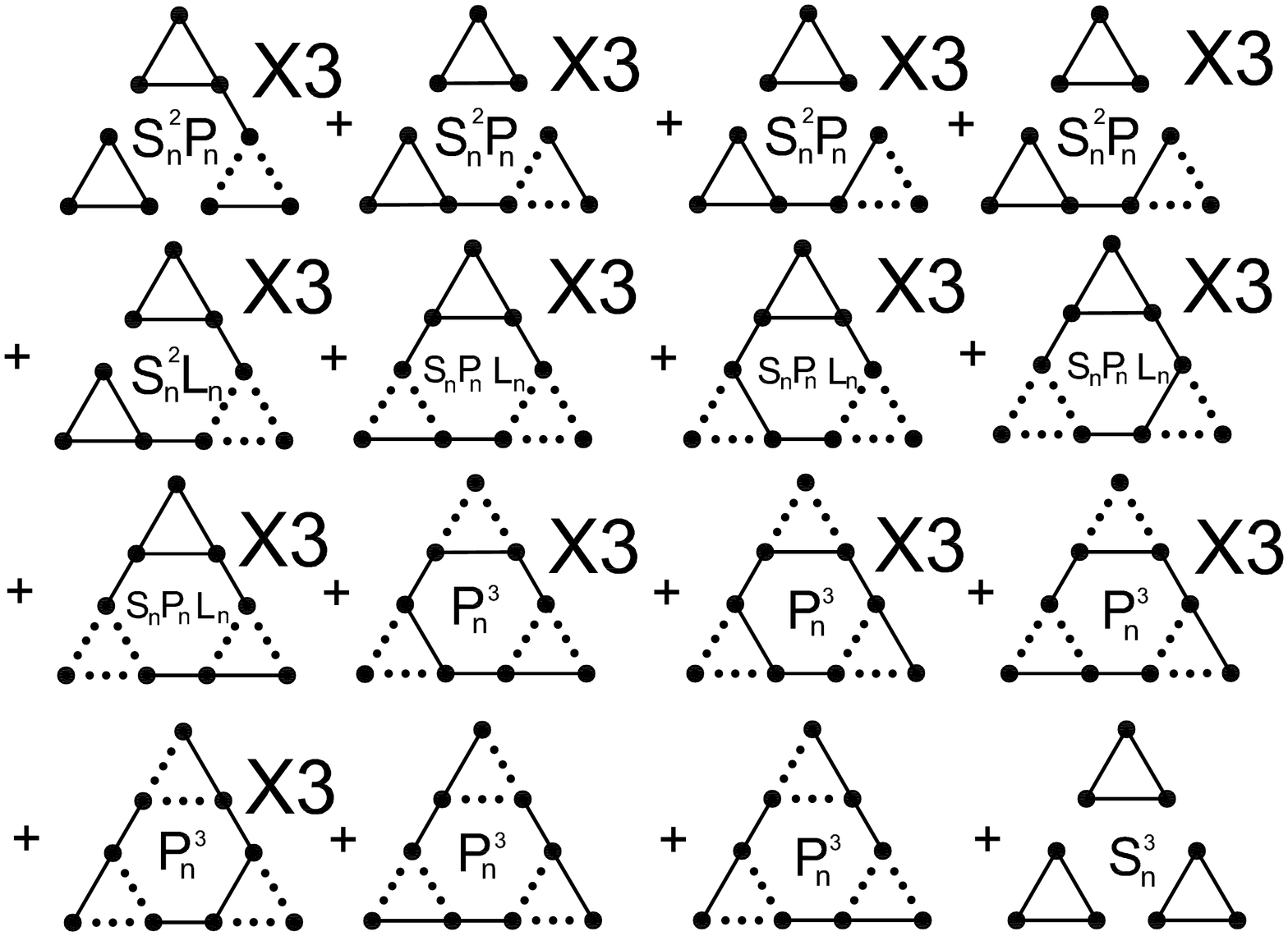}
\end{center}
\caption{Spanning subgraphs of $H_{n+1}$ that contribute to the term
$s^3_{n}+12s^2_np_n+3s^2_nl_n+12s_np_nl_n+14p^3_n$ of $l_{n+1}$.}
\label{ln}
\end{figure}
\begin{figure}[ht]
\begin{center}
\includegraphics[width=0.6\textwidth,angle=90,trim=20 100 0 0]{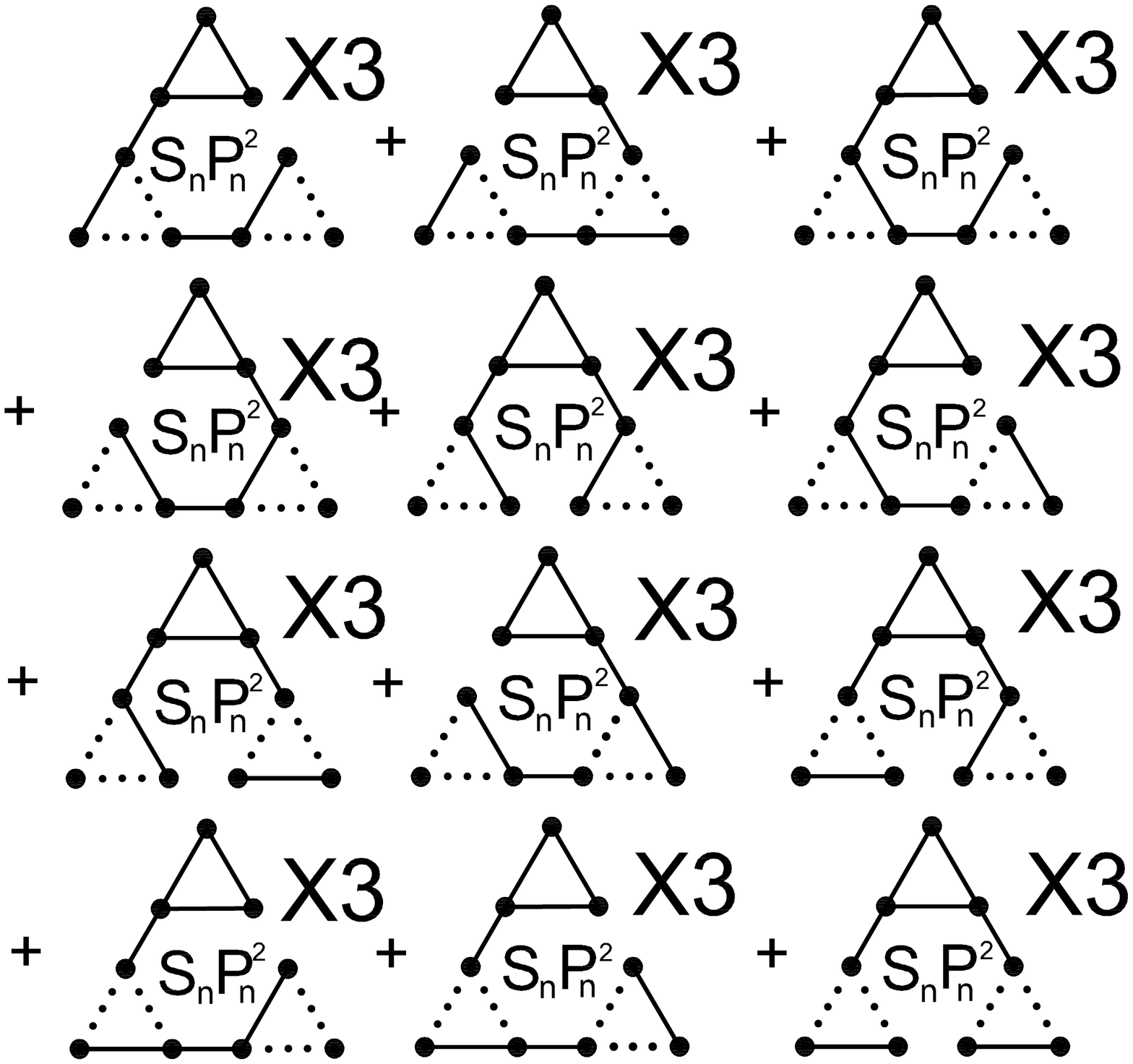}
\end{center}
\caption{Spanning subgraphs of $H_{n+1}$ that contribute to the term
$36s^2_{n}p_{n}$ of $l_{n+1}$.} \label{lnplus}
\end{figure}

\begin{lemma} \label{th:fnpnln}
For  the Hanoi graph $H_n$ with $n \geq 1$, $s_n l_n=3p_n^2$.
\end{lemma}
\begin{proof}
By induction.
For $n=1$, using the initial conditions $s_1=3$, $p_1=1$ and $l_1=1$, the result is true.
Let us assume that for $n=k$, the lemma is true.
For $n=k+1$,  using Lemma~\ref{th:fn}, we have that
\begin{align*}
s_{k+1}l_{k+1}-3p_{k+1}^2&= (3s_k^3+6s_k^2p_k)(s_k^3+12s_k^2p_k+3s_k^2l_k+36s_kp_k^2+12s_kp_kl_k\\
&\quad +14p_k^3)-3(s_k^3+7s_k^2p_k+7s_kp_k^2+s_k^2l_k)^2\\
&= 3s_k^2(s_k^2+4s_kp_k+7p_k^2-s_kl_k)(s_kl_k-3p_k^2)\,.
\end{align*}
By induction hypothesis  $s_{k}l_{k}-3p_{k}^2=0$, we obtain the result.
\end{proof}

\begin{lemma} \label{th:fnfn}
For  the Hanoi graph $H_n$ with $n \geq 1$, $\frac{s_{n+1}}{s_n^3}=\frac{5^n}{3^{n-1}}$.
\end{lemma}
\begin{proof}
From Eq.~\eqref{eq:sna}, we have
\begin{equation*}
\frac{s_{n+1}}{s_n^3}=\frac{3s_n^3+6s_n^2p_n}{s_n^3}=3+6\frac{p_n}{s_n},
\end{equation*}
which can be rewritten as
\begin{equation*}
\frac{p_n}{s_n}=\frac{1}{6}\left(\frac{s_{n+1}}{s_n^3}-3\right).
\end{equation*}
Using Eq.~\eqref{eq:snb} and Lemma~\ref{th:fnpnln}, we obtain
\begin{align*}
\frac{p_{n+1}}{s_n^3}&= 1+7\frac{p_n}{s_n}+10\left(\frac{p_n}{s_n}\right)^2 \\
&= 1+7\left[\frac{1}{6}\left(\frac{s_{n+1}}{s_n^3}-3\right)\right]+10\left[\frac{1}{6}\left(\frac{s_{n+1}}{s_n^3}-3\right)\right]^2\\
&= -\frac{s_{n+1}}{2s_n^3}+\frac{5s_{n+1}^2}{18s_n^6},
\end{align*}
which leads to
\begin{equation*}
\frac{p_{n+1}}{s_{n+1}}=\frac{p_{n+1}}{s_n^3}\frac{s_n^3}{s_{n+1}}=\left(-\frac{s_{n+1}}{2s_n^3}+\frac{5s_{n+1}^2}{18s_n^6}\right)\frac{s_n^3}{s_{n+1}}=-\frac{1}{2}+\frac{5s_{n+1}}{18s_n^3}.
\end{equation*}
According to Eq.~\eqref{eq:sna}, we have $s_{n+2}=3s^3_{n+1}+6s^2_{n+1} p_{n+1}$ and
\begin{equation*}
\frac{s_{n+2}}{s_{n+1}^3}=3+6\frac{p_{n+1}}{s_{n+1}}=3+6\left(-\frac{1}{2}+\frac{5s_{n+1}}{18s_n^3}\right)=\frac{5s_{n+1}}{3s_n^3},
\end{equation*}
which, together with the initial condition $\frac{s_2}{s_1^3}=5$ yields
\begin{equation*}
\frac{s_{n+1}}{s_n^3}=\frac{5^n}{3^{n-1}}\,.
\end{equation*}
\end{proof}

We now give one of the main results of this paper.
{
\begin{theorem}\label{th:solve}
For the Hanoi graph $H_n$, with $n\geq 1$,  the number of spanning trees $s_n$ and spanning subgraphs
$p_n$ and $l_n$  is
\begin{equation}\label{so:sna}
s_n=3^{\frac{1}{4}3^n+\frac{1}{2}n-\frac{1}{4}}\cdot
5^{\frac{1}{4}3^n-\frac{1}{2}n-\frac{1}{4}},
\end{equation}
\begin{equation}\label{so:snb}
p_n=\frac{1}{6}\cdot\frac{5^n-3^n}{5^n}\cdot3^{\frac{1}{4}3^n-\frac{1}{2}n+\frac{3}{4}}\cdot5^{\frac{1}{4}3^n+\frac{1}{2}n-\frac{1}{4}},
\end{equation}
\begin{equation}\label{so:snc}
l_n=\frac{1}{4}\cdot\left( {3^n-5^n}\right)^2\cdot
3^{\frac{1}{4}3^n-\frac{3}{2}n+\frac{3}{4}}\cdot5^{\frac{1}{4}3^n-\frac{1}{2}n-\frac{1}{4}}\,.
\end{equation}
\end{theorem}}
\begin{proof}
From Lemma \ref{th:fnfn}, we have  $s_{n+1}=\frac{5^n}{3^{n-1}}s_n^3$, which
with initial condition $s_1=3$ gives Eq.~\eqref{so:sna}.

From the proof of Lemma~\ref{th:fnfn} we know that $p_n=\frac{s_{n+1}-3s_n^3}{6s_n^2}$. 
Inserting the expressions for $s_{n+1}$ and $s_n$ in Eq.~\eqref{so:sna} into this formula leads to $p_n$.

Lemma~\ref{th:fnpnln} gives $l_n=\frac{3p_n^2}{s_n}$.
Using the obtained results for $s_n$ and $p_n$, we arrive at
$l_n$.
\end{proof}

Note that Eq.~\eqref{so:sna} was previously obtained~\cite{TeWa11} by using a different method.

After finding an explicit expression for the number of spanning trees
of $H_n$, we now calculate its spanning tree  entropy which is
defined as:
\begin{equation}
\label{eq:en} h=\lim_{V_n \to \infty}\frac{s_n}{V_n}
\end{equation}
where $V_n$ denotes the number of vertices, see~\cite{Ly05}.

Thus, for the Hanoi graph  we obtain $h=\frac{1}{4}(\ln3+\ln5)\simeq
0.677$.

We can compare this asymptotic value of the entropy of  the spanning
trees of $H_n$ with those of other graphs with the same average
degree. For example, the value for the honeycomb lattice is
0.807~\cite{Wu77} and  the 4-8-8 (bathroom tile)  and 3-12-12
lattices have  entropy values 0.787 and 0.721,
respectively~\cite{ShWu00}. Thus, the asymptotic value for the Hanoi
graph is the lowest reported for graphs with average degree 3. This
reflects the fact that the number of spanning trees in $H_n$,
although growing exponentially, does so at a lower rate than lattices
with the same average degree.

\section{The degree distribution for a vertex of the spanning trees}

In this section, we compute the probabilities of different coordination numbers at a given vertex on a random spanning tree on the Hanoi graph $H_n$. We note that, by using similar techniques, it has been possible to obtain more results for the closely related Siperpi\'nski graphs~\cite{ChCh10,ShTeWa14}.
In the previous calculation, each vertex of $H_n$ corresponds to a state/configuration of all $n$ disks and thus is labeled by an $n$-tuple $\bm{\alpha}=\alpha_1\alpha_2\cdots\alpha_n$.

In what follows for the convenience of description,  we provide an alternative way of labeling vertices in $H_n$, by assigning to each vertex a sequence $\bm{\alpha}=\alpha_1\alpha_2\cdots\alpha_e$, where
$1\leq e \leq n$ and $\alpha_i\in\{0,1,2\}$. The new labeling method is as follows, see Fig.~\ref{hanchan}. For $n=1$, $H_1$ is a triangle, we label the three vertices by $0$, $1$ and $2$. When $n=2$,
$H_2$ contains three replicas of  $H_1$, denoted by $H_1^{1}$, $H_1^{2}$, and $H_1^{3}$. On the topmost copy $H_1^{1}$, we put a prefix
$0$ on the label of each node in $H_1$. Similarly, we add a prefix
$1$ (or $2$) to the labeling of vertices on the leftmost (or rightmost) copy $H_1^{2}$ (or $H_1^{3}$). If a vertex's label ends with several
identical digits, we just keep it once. For example, we use
$010$ to replace $0100$. For $n\geq 3$, we label the vertices in $H_{n}$ by adding
prefixes to three replicas of $H_{n-1}$ in the same way, and delete repetitive suffix.

\begin{figure}[htb]
\begin{center}
\includegraphics[width=0.6\textwidth,angle=0,trim=80 60 160 20]{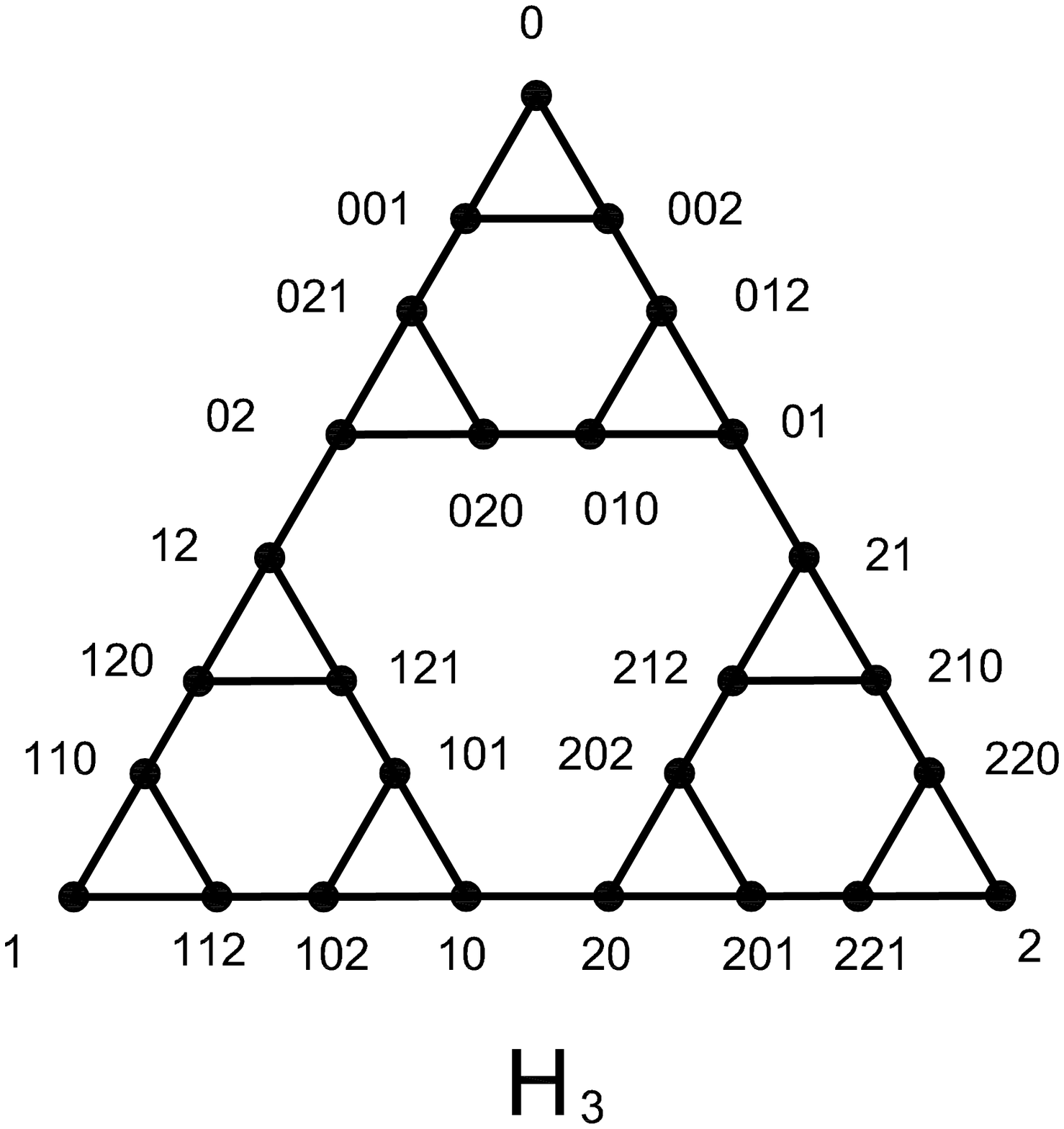}
\end{center}
\caption{An illustration for a new labeling of vertices in $H_3$.}\label{hanchan}
\end{figure}

In this way, all vertices in $H_{n}$ are labeled by sequences of three digits 0, 1, and 2, with different length ranging from 1 to $n$, and each vertex has a unique labeling. For example, for all $n$, the three outmost vertices of $H_{n}$ have labels of 0, 1, and 2, while the other six outmost vertices of $H_{n-1}^{1}$, $H_{n-1}^{2}$, and $H_{n-1}^{3}$ forming $H_{n}$, each has a label consisting of two digits, which are called connecting vertices hereafter.

After labeling the vertices in $H_{n}$, we are now in a position to study the probability distribution of degree for a vertex on all spanning trees. For this purpose, we introduce some quantities.

\begin{definition}
Consider a vertex $\bm{\alpha}$ in $H_n$. We define
$s_{n,i}(\bm{\alpha})$ as the number of spanning trees in which the
degree of the node $\bm{\alpha}$ is $i$. 
Then the probability that among all spanning tree the degree of vertex $\bm{\alpha}$ is $i$ is defined by
$S_{n,i}(\bm{\alpha})=s_{n,i}(\bm{\alpha})/s_n$. 
{Similarly, we
define $r_{n,i}(\bm{\alpha})$ $(t_{n,i}(\bm{\alpha})$, $p_{n,i}(\bm{\alpha}))$ as the number of spanning subgraphs
consisting of two trees such that one outmost vertex $0$ $(1$, $2)$ is in one tree while the other two outmost vertices  
$1$ and $2$ $(0$ and $2$, $0$ and $1)$   are in the other tree, and the degree of $\bm{\alpha}$ is $i$.
Define the probabilities $R_{n,i}(\bm{\alpha})=r_{n,i}(\bm{\alpha})/r_n$, $T_{n,i}(\bm{\alpha})=t_{n,i}(\bm{\alpha})/t_n$ $P_{n,i}(\bm{\alpha})=p_{n,i}(\bm{\alpha})/p_n$.} 
Finally, we define $l_{n,i}(\bm{\alpha})$ as the number of spanning subgraphs containing
three trees such that the three outmost vertices $0$, $1$ and $2$ belongs to
a different tree, and the degree of $\bm{\alpha}$ is $i$.
Define the probability $L_{n,i}(\bm{\alpha})=l_{n,i}(\bm{\alpha})/l_n$.
\end{definition}

In the following text, we will first determine $S_{n,i}(\bm{\alpha})$ for
the three outmost vertices in in $H_n$, then we will compute $S_{n,i}(\bm{\alpha})$ for the six connecting vertices, and finally we will calculate $S_{n,i}(\bm{\alpha})$ for an arbitrary  vertex $\bm{\alpha}$.

For the three outmost vertices 0, 1, and 2, each has a degree of 2, and thus $s_{n,3}(0)=p_{n,3}(0)=l_{n,3}(0)=0$. In addition, by symmetry we have
$s_{n,i}(0)=s_{n,i}(1)=s_{n,i}(2)$, $p_{n,i}(0)=p_{n,i}(1)$ for $i=1,2$, and $l_{n,i}(0)=l_{n,i}(1)=l_{n,i}(2)$ for $i=0,1,2$. Hence, for the outmost vertices, we only need to find $S_{n,i}(0)$ for $i=1,2$.

\subsection{Determination of $S_{n,i}(0)$ with $i=1,2$}

For the graph $H_n$, associated with the Tower of Hanoi puzzle with  $n$ disc, we have the following result.
{\begin{theorem}\label{ther1} For the Hanoi graph $H_n$ with $n\ge1$,
\begin{equation} \label{sn120a}
S_{n,1}(0) = \frac{5}{7}-\frac{5}{7}\left(\frac{1}{15}\right)^n\,,
\end{equation}
\begin{equation} \label{sn120b}
S_{n,2}(0)=\frac{2}{7}+\frac{5}{7}\left(\frac{1}{15}\right)^n\,.
\end{equation}
\begin{equation} \label{pn120a}
P_{n,1}(0) = \frac{5}{7}+\frac{9\cdot 5^n+5\cdot 3^n}{7\cdot15^n\cdot(5^n-3^n)}\,,
\end{equation}
\begin{equation} \label{pn120b}
P_{n,2}(0)=\frac{2}{7}-\frac{9\cdot 5^n+5\cdot
                  3^n}{7\cdot15^n\cdot(5^n-3^n)}\,,
\end{equation}
\begin{equation} \label{Pn0n12a}
P_{n,0}(2)=\frac{5\cdot(15^n-1)}{7\cdot5^n\cdot(5^n-3^n)} \,,
\end{equation}
\begin{equation} \label{Pn0n12b}
P_{n,1}(2)=\frac{5}{7}-\frac{12}{7}\left(\frac{1}{15}\right)^n-\frac{3\cdot(15^n-1)}{7\cdot5^n\cdot(5^n-3^n)}\,,
\end{equation}
\begin{equation} \label{Pn0n12c}
P_{n,2}(2)=\frac{2}{7}+\frac{12}{7}\left(\frac{1}{15}\right)^n-\frac{2\cdot(15^n-1)}{7\cdot5^n\cdot(5^n-3^n)}\,,
\end{equation}
\begin{equation} \label{Hn0n1n20a}
L_{n,0}(0)=\frac{10\cdot3^n}{21\cdot(5^n-3^n)}+\frac{18\cdot5^n+10\cdot3^n}{21\cdot5^n\cdot(5^n-3^n)^2}\,, \end{equation}
\begin{equation} \label{Hn0n1n20b}
L_{n,1}(0)=\frac{5}{7}+\frac{9}{7}\left(\frac{1}{15}\right)^n-\frac{2\cdot(15^n-1)}{7\cdot5^n\cdot(5^n-3^n)}-\frac{8\cdot3^n}{3\cdot5^n\cdot(5^n-3^n)^2}\,,
\end{equation}
\begin{equation} \label{Hn0n1n20c}
L_{n,2}(0)=\frac{2}{7}-\frac{9}{7}\left(\frac{1}{15}\right)^n-\frac{4\cdot(15^n+6)}{21\cdot5^n\cdot(5^n-3^n)}+\frac{4\cdot3^n}{3\cdot5^n\cdot(5^n-3^n)^2}\,.
\end{equation}
\end{theorem}
}

In order to prove Theorem~\ref{ther1}  and other main results, we shall first give the following lemma.
\begin{lemma}\label{Lemma01}
For the Tower of Hanoi graph $H_n$ with $n\geq1$,
\begin{equation} \label{sn1sn2a}
s_{n,1}(0)=\left(\frac{5}{7}-\frac{5}{7}\left(\frac{1}{15}\right)^n\right)\cdot3^{\frac{1}{4}3^n+\frac{1}{2}n-\frac{1}{4}} \cdot 5^{\frac{1}{4}3^n-\frac{1}{2}n-\frac{1}{4}}\,,
\end{equation}
\begin{equation} \label{sn1sn2b}
s_{n,2}(0)=\left(\frac{2}{7}+\frac{5}{7}\left(\frac{1}{15}\right)^n\right)\cdot3^{\frac{1}{4}3^n+\frac{1}{2}n-\frac{1}{4}} \cdot 5^{\frac{1}{4}3^n-\frac{1}{2}n-\frac{1}{4}}\,,
\end{equation}
\begin{equation} \label{pn1n2sola}
p_{n,1}(0)=\left(\frac{5}{14}\frac{5^n-3^n}{5^n}+\frac{9\cdot5^n+5\cdot3^n}{14\cdot75^n}\right)\cdot3^{\frac{1}{4}3^n-\frac{1}{2}n-\frac{1}{4}}\cdot5^{\frac{1}{4}3^n+\frac{1}{2}n-\frac{1}{4}}\,, \end{equation}
\begin{equation} \label{pn1n2solb}
p_{n,2}(0)=\left(\frac{1}{7}\frac{5^n-3^n}{5^n}-\frac{9\cdot5^n+5\cdot3^n}{14\cdot75^n}\right)\cdot3^{\frac{1}{4}3^n-\frac{1}{2}n-\frac{1}{4}}\cdot5^{\frac{1}{4}3^n+\frac{1}{2}n-\frac{1}{4}}\,,
\end{equation}
\begin{equation} \label{pn02sola}
p_{n,0}(2) =\frac{1}{14}\left(1-\frac{1}{15^n}\right)\cdot3^{\frac{1}{4}3^n+\frac{1}{2}n-\frac{1}{4}}\cdot5^{\frac{1}{4}3^n-\frac{1}{2}n+\frac{3}{4}}\,,
\end{equation}
\begin{equation} \label{pn02solb}
p_{n,1}(2)=\frac{5}{14}\left(1-\frac{8\cdot3^n}{5\cdot5^n}-\frac{12}{5\cdot15^n}+\frac{3}{25^n}\right)\cdot3^{\frac{1}{4}3^n-\frac{1}{2}n-\frac{1}{4}}\cdot5^{\frac{1}{4}3^n+\frac{1}{2}n-\frac{1}{4}}\,,
\end{equation}
\begin{equation} \label{pn02solc}
p_{n,2}(2)=\frac{12}{7}\left(1-\frac{2\cdot3^n}{5^n}+\frac{6}{15^n}-\frac{5}{25^n}\right)\cdot3^{\frac{1}{4}3^n-\frac{1}{2}n-\frac{1}{4}}\cdot5^{\frac{1}{4}3^n+\frac{1}{2}n-\frac{1}{4}}\,,
\end{equation}
\begin{equation} \label{ln02sola}
l_{n,0}(0) = \frac{5}{14}\left(\frac{5^n-3^n}{5^n}+\frac{9}{5\cdot15^n}+\frac{1}{25^n}\right)\cdot3^{\frac{1}{4}3^n-\frac{1}{2}n-\frac{1}{4}} \cdot5^{\frac{1}{4}3^n+\frac{1}{2}n-\frac{1}{4}}\,,
\end{equation}
\begin{small}
\begin{equation} \label{ln02solb}
l_{n,1}(0)=\frac{15}{28}\left(1-\frac{12\cdot3^n}{5\cdot5^n}+\frac{7\cdot9^n}{5\cdot25^n}+\frac{9}{5\cdot15^n}-\frac{16}{5\cdot25^n}\right)
 3^{\frac{1}{4}3^n-\frac{3}{2}n-\frac{1}{4}}
                  \cdot5^{\frac{1}{4}3^n+\frac{3}{2}n-\frac{1}{4}}\,,
\end{equation}
\end{small}
\begin{footnotesize}
\begin{equation} \label{ln02solc}
l_{n,2}(0)=\frac{3}{14}\left(1-\frac{8\cdot3^n}{3\cdot5^n}+\frac{5\cdot9^n}{3\cdot25^n}-\frac{9}{2\cdot15^n}+\frac{5}{25^n}+\frac{25\cdot3^n}{6\cdot125^n}\right)3^{\frac{1}{4}3^n-\frac{3}{2}n-\frac{1}{4}}
                  \cdot5^{\frac{1}{4}3^n+\frac{3}{2}n-\frac{1}{4}}.
\end{equation}
\end{footnotesize}
\end{lemma}
\begin{proof}
Based on Figs.~\ref{fn},~\ref{pn},~\ref{ln}, and~\ref{lnplus}, we can establish the following recursive relations
\begin{equation} \label{sni0a}
s_{n+1,i}(0) = 3s_{n,i}(0)s_n^2+2p_{n,i}(0)s_n^2+4s_{n,i}(0)s_np_n \,,
\end{equation}
\begin{small}
\begin{equation} \label{sni0b}
p_{n+1,i}(0)=s_{n,i}(0)s_n^2+p_{n,i}(0)s_n^2+6s_{n,i}(0)s_np_n+4p_{n,i}(0)s_np_n+3s_{n,i}(0)p_n^2+s_{n,i}(0)p_nl_n\,,
\end{equation}
\end{small}
\begin{align} \label{sni0c}
p_{n+1,i}(2)&= s_{n,i}(0)s_n^2+2p_{n,i}(0)s_n^2+3p_{n,i}(2)s_n^2+l_{n,i}(0)s_n^2+ \nonumber \\
&\quad 2s_{n,i}(0)s_np_n+2p_{n,i}(0)s_np_n+4p_{n,i}(2)s_np_n+s_{n,i}(0)p_n^2 \,,
\end{align}
\begin{align} \label{sni0d}
l_{n+1,i}(0)&=  s_{n,i}(0)s_n^2+2p_{n,i}(0)s_n^2+2p_{n,i}(2)s_n^2+l_{n,i}(0)s_n^2+ \nonumber\\
&\quad  8s_{n,i}(0)s_np_n+12p_{n,i}(0)s_np_n+12p_{n,i}(2)s_np_n+\nonumber\\
&\quad 4l_{n,i}(0)s_np_n+12s_{n,i}(0)p_n^2+8p_{n,i}(0)p_n^2+6p_{n,i}(2)p_n^2+\nonumber\\
&\quad 2s_{n,i}(0)s_nl_n+2p_{n,i}(0)s_nl_n+2p_{n,i}(2)s_nl_n+4s_{n,i}(0)p_nl_n\,.
\end{align}
Using the initial conditions
$s_{1,1}(0)=2$, $s_{1,2}(0)=1$, $p_{1,1}(0)=1$, $p_{1,2}(0)=0$, $p_{1,0}(2)=1$, $p_{1,1}(2)=p_{1,2}(2)=0$, $l_{1,0}(0)=1$, and
$l_{1,1}(0)=l_{1,2}(0)=0$, the above recursive relations are solved to obtain Lemma~\ref{Lemma01}.
\end{proof}

From Eqs.~(\ref{so:sna}-\ref{so:snc}) and Eqs.~(\ref{sn1sn2a}-\ref{ln02solc}), we can prove Theorem~\ref{ther1}.

\subsection{Determination of $S_{n,i}(\bm{\alpha})$ with $\bm{\alpha}$ being connecting vertices}

We proceed to calculate $S_{n,i}(\bm{\alpha})$ with $i=1,2,3$,
where $\bm{\alpha}$ are the six connecting vertices, the length of whose labels is two. By definition, the six connecting vertices are  01, 10, 02, 20, 12, 21. We obtain that $S_{n,i}(01)=S_{n,i}(02)=S_{n,i}(10)=S_{n,i}(12)=S_{n,i}(20)=S_{n,i}(21)$. Since connecting vertices only exist in $H_n$ for $n\geq 2$, we only need to determine  $S_{n+1,1}(01)$ for $n\geq 1$. Thus,
{
\begin{theorem}
For the Tower of Hanoi graph $H_n$ and $n\geq1$,
\begin{equation} \label{Sn1101a}
S_{n+1,1}(01) = \frac{1}{14}5^{1-2n}\cdot(15^n-1) \,,
\end{equation}
\begin{equation} \label{Sn1101b}
S_{n+1,2}(01)=\frac{1}{42}\left(30+6\cdot5^{1-2n}-3^{2+n}\cdot5^{-n}-23\cdot5^{-n}\right) \,,
\end{equation}
\begin{equation} \label{Sn1101c}
S_{n+1,3}(01)=\frac{1}{42}\left(12-3\cdot5^{1-2n}-2\cdot3^{1+n}\cdot5^{-n}+23\cdot5^{-n}\right),
\end{equation}
\begin{equation} \label{pn120a}
P_{n+1,1}(01)= \frac{3\cdot5^{1-2 n} \left(5^n-3^n\right) \left(15^n-1\right)}{14\left(5^{n+1}-3^{n+1}\right)}\,,
\end{equation}
\begin{equation} \label{pn120b}
P_{n+1,2}(01)=\frac{7\left(\frac{9}{5}\right)^n+5\cdot3^{1-n}-2\cdot3^n-19\cdot5^{-n}-3\cdot5^{n+1}}{7\cdot3^{n+1}-7\cdot5^{n+1}}\,,
\end{equation}
\begin{small}
\begin{equation} \label{pn120c}
P_{n+1,3}(01)=\frac{49\cdot3^{1-n}-28\cdot3^{n+1}+3^n 5^{3-2 n}-56\cdot5^{1-n}+42\cdot5^n+2\cdot 5^{2-n}9^n}{14 \left(5^{n+1}-3^{n+1}\right)},
\end{equation}
\end{small}
\begin{equation} \label{hn120a}
P_{n+1,1}(02) = \frac{3\cdot5^{1-2 n} \left(3\cdot5^n-3^n\right) \left(15^n-1\right)}{14\left(5^{n+1}-3^{n+1}\right)} \,,
\end{equation}
\begin{footnotesize}
\begin{equation} \label{hn120b}
P_{n+1,2}(02)=\frac{19\cdot3^{1-n}+19\cdot3^{n+1}+2\cdot3^{n+1} 5^{1-2 n}-113\cdot5^{-n}-2\cdot5^{n+2}-5^{-n}9^{n+1}}{14 \left(3^{n+1}-5^{n+1}\right)}\,,
\end{equation}
\begin{equation} \label{hn120c}
P_{n+1,3}(02)=\frac{75^{-n} \left(106\cdot3^n 5^{n+1}-125\cdot9^n-453\cdot25^n-2\cdot5^{n+2} 27^n+184\cdot225^n-86\cdot375^n\right)}{14 \left(3^{n+1}-5^{n+1}\right)},
\end{equation}
\end{footnotesize}
\begin{equation} \label{h4n420a}
P_{n+1,1}(20) = \frac{25^{-n} \left(5\cdot3^n+21\cdot5^n+7\cdot3^n 5^{2 n+1}-5^{n+1} 9^n\right)}{14\left(5^{n+1}-3^{n+1}\right)}\,,
\end{equation}
\begin{footnotesize}
\begin{equation} \label{h4n420b}
P_{n+1,2}(20)=\frac{55\cdot3^{-n}-19\cdot3^{n+1}+3^n 5^{1-2 n}-11\cdot5^{1-n}+2\cdot5^{n+2}+5^{-n}9^{n+1}}{14 \left(5^{n+1}-3^{n+1}\right)}\,,
\end{equation}
\begin{equation} \label{h4n420c}
P_{n+1,3}(20)=\frac{75^{-n} \left(95\cdot9^n-38\cdot15^n-77\cdot25^n-14\cdot3^{n+1} 125^n+38\cdot135^n+52\cdot225^n\right)}{14 \left(3^{n+1}-5^{n+1}\right)},
\end{equation}
\end{footnotesize}
\begin{footnotesize}
\begin{equation} \label{h46u420a}
L_{n+1,1}(01) = \frac{25^{-n} \left(37\cdot3^n 5^{3 n+1}-25\cdot9^n+38\cdot15^n+39\cdot25^n-2\cdot3^{2 n+1}25^{n+1}+5^{n+2} 27^n\right)}{14 \left(3^{n+1}-5^{n+1}\right)^2}\,,
\end{equation}
\begin{align} \label{h46u420b}
L_{n+1,2}(01)&=\frac{75^{-n} \left(-3^{4n+3} 5^n+27\cdot5^{3 n+1}-29\cdot3^{2 n+1} 5^{3 n+1}+2\cdot3^n 5^{4n+3}\right)}{14\left(3^{n+1}-5^{n+1}\right)^2}\nonumber\\
&\quad +75^{-n}\frac{20\cdot27^n+8\cdot25^n 27^{n+1}-13\cdot45^n-188\cdot75^n}{14\left(3^{n+1}-5^{n+1}\right)^2}\,,
\end{align}
\begin{align} \label{h46u420c}
L_{n+1,3}(01)&=\frac{75^{-n} \left(-319\cdot3^{n+1} 25^n+62\cdot3^{3 n+1} 25^n+65\cdot27^n\right)}{14\left(3^{n+1}-5^{n+1}\right)^2}\nonumber\\
&\quad  +75^{-n}\frac{199\cdot45^n+789\cdot125^n+26\cdot405^n-562\cdot1125^n+254\cdot1875^n}{14\left(3^{n+1}-5^{n+1}\right)^2}.
\end{align}
\end{footnotesize}
\end{theorem}
}

\begin{proof}
Based on Fig.~\ref{fn}, we have the recursion relations for the connecting vertex 01:
\begin{equation} \label{length1a}
s_{n+1,1}(01) = s_{n,1}(1)s_n^2+p_{n,0}(2)s_n^2 \,,
\end{equation}
\begin{equation} \label{length1b}
s_{n+1,2}(01)=s_{n,2}(1)s_n^2+2s_{n,1}(1)s_n^2+4s_{n,1}(1)s_np_n+p_{n,1}(1)s_n^2+p_{n,1}(2)s_n^2 \,,
\end{equation}
\begin{equation} \label{length1c}
s_{n+1,3}(01)=2s_{n,2}(1)s_n^2+4s_{n,2}(1)s_np_n+p_{n,2}(1)s_n^2+p_{n,2}(2)s_n^2\,.
\end{equation}
Since the quantities on the right-hand side of Eqs.~(\ref{length1a}-\ref{length1c}) have been explicitly determined, according to the relation $S_{n+1,i}(01)=s_{n+1,i}(01)/s_{n+1}$, we obtain Eqs.~(\ref{Sn1101a}-\ref{Sn1101c}).

Analogously, we find $P_{n+1,i}(\bm{\alpha})$ and $L_{n+1,i}(\bm{\alpha})$, when $\bm{\alpha}$ are connecting vertices. It is obvious that
$P_{n,0}(\bm{\alpha})=L_{n,0}(\bm{\alpha})=0$. Note that $l_{n,i}(01)=l_{n,i}(02)=l_{n,i}(10)=l_{n,i}(12)=l_{n,i}(20)=l_{n,i}(21)$,
$p_{n,i}(01)=p_{n,i}(10)$, $p_{n,i}(02)=p_{n,i}(12)$,
$p_{n,i}(20)=p_{n,i}(21)$. Using Figs.~\ref{pn},~\ref{ln}, and~\ref{lnplus}, we can establish the recursive relations
\begin{equation} \label{length2a}
p_{n+1,1}(01) = s_{n,1}(1)s_np_n+p_{n,0}(2)s_np_n \,,
\end{equation}
\begin{align} \label{length2b}
p_{n+1,2}(01)&= s_{n,2}(1)s_np_n+s_{n,1}(1)s_n^2+5s_{n,1}(1)s_np_n+s_{n,1}(1)s_nl_n+ \nonumber\\
&\quad 3s_{n,1}p_n^2+p_{n,1}(1)s_n^2+3p_{n,1}(1)s_np_n+p_{n,1}(2)s_np_n \,,
\end{align}
\begin{align} \label{length2c}
p_{n+1,3}(01)&= s_{n,2}(1)s_n^2+5s_{n,2}(1)s_np_n+s_{n,2}(1)s_nl_n+3s_{n,2}(1)p_n^2 \nonumber \\
&\quad  p_{n,2}(1)s_n^2+3p_{n,2}(1)s_np_n+p_{n,2}(2)s_np_n\,,
\end{align}
\begin{equation} \label{length3a}
p_{n+1,1}(02) = s_{n,1}(1)s_n^2+3s_{n,1}(1)s_np_n+p_{n,0}(2)s_n^2+3p_{n,0}(2)s_np_n \,,
\end{equation}
\begin{align} \label{length3b}
p_{n+1,2}(02)&= s_{n,2}(1)s_n^2+3s_{n,2}(1)s_np_n+3s_{n,1}(1)s_np_n+s_{n,1}(1)s_nl_n+ \nonumber\\
&\quad  3s_{n,1}(1)p_n^2+p_{n,1}(2)s_n^2+3p_{n,1}(2)s_np_n+p_{n,1}(1)s_np_n \,,
\end{align}
\begin{align} \label{length3c}
p_{n+1,3}(02)&= 3s_{n,2}(1)s_np_n+s_{n,2}(1)s_nl_n+3s_{n,2}(1)p_n^2+p_{n,2}(2)s_n^2+\nonumber\\
&\quad  3p_{n,2}(2)s_np_n+p_{n,2}(1)s_np_n \,,
\end{align}
\begin{equation} \label{length4a}
p_{n+1,1}(20) = s_{n,1}(1)s_n^2+s_{n,1}(1)s_np_n+2p_{n,1}(1)s_n^2+p_{n,0}(2)s_n^2+p_{n,0}(2)s_np_n \,,
\end{equation}
\begin{align} \label{length4b}
p_{n+1,2}(20)&= s_{n,2}(1)s_n^2+s_{n,2}(1)s_np_n+2p_{n,2}(1)s_n^2+s_{n,1}(1)s_np_n+\nonumber\\
&\quad  s_{n,1}(1)p_n^2+2p_{n,1}(1)s_n^2+5p_{n,1}(1)s_np_n+p_{n,1}(2)s_n^2+\nonumber\\
&\quad  p_{n,1}(2)s_np_n+l_{n,1}(1)s_n^2  \,,
\end{align}
\begin{align} \label{length4c}
p_{n+1,3}(20)&= s_{n,2}(1)s_np_n+s_{n,2}(1)p_n^2+2p_{n,2}(1)s_n^2+5p_{n,2}(1)s_np_n+\nonumber\\
&\quad
p_{n,2}(2)s_n^2+p_{n,2}(2)s_np_n+l_{n,2}(1)s_n^2 \,,
\end{align}
\begin{align} \label{length5a}
l_{n+1,1}(01) &= s_{n,1}(1)s_n^2+6s_{n,1}(1)s_np_n+s_{n,1}(1)s_nl_n+4s_{n,1}(1)p_n^2+\nonumber\\
&\quad  2p_{n,1}(1)s_n^2+8p_{n,1}s_np_n+p_{n,0}(2)s_n^2+6p_{n,0}(2)s_np_n+\nonumber\\
&\quad  p_{n,0}(2)s_nl_n+4p_{n,0}(2)p_n^2 \,,
\end{align}
\begin{align} \label{length5b}
l_{n+1,2}(01)&= s_{n,2}(1)s_n^2+6s_{n,2}(1)s_np_n+s_{n,2}(1)s_nl_n+4s_{n,2}(1)p_n^2+\nonumber\\
&\quad  2p_{n,2}(1)s_n^2+8p_{n,2}(1)s_np_n+2s_{n,1}(1)s_np_n+s_{n,1}(1)s_nl_n+\nonumber\\
&\quad  8s_{n,1}(1)p_n^2+4s_{n,1}(1)p_nl_n+p_{n,1}(1)s_n^2+10p_{n,1}(1)s_np_n+\nonumber\\
&\quad   3p_{n,1}(1)s_nl_n+10p_{n,1}(1)p_n^2+p_{n,1}(2)s_n^2+6p_{n,1(2)}s_np_n+\nonumber\\
&\quad  p_{n,1}(2)s_nl_n+4p_{n,1}(2)p_n^2+l_{n,1}(1)s_n^2+4l_{n,1}(1)s_np_n \,,
\end{align}
\begin{align} \label{length5c}
l_{n+1,3}(01)&=  2s_{n,2}(1)s_np_n+s_{n,2}(1)s_nl_n+8s_{n,2}(1)p_n^2+4s_{n,2}(1)p_nl_n+\nonumber\\
&\quad p_{n,2}(1)s_n^2+10p_{n,2}(1)s_np_n+3p_{n,2}(1)s_nl_n+10p_{n,2}(1)p_n^2+\nonumber\\
&\quad p_{n,2}(2)s_n^2+6p_{n,2}(2)s_np_n+p_{n,2}(2)s_nl_n+4p_{n,2}(2)p_n^2+\nonumber\\
&\quad l_{n,2}(1)s_n^2+4l_{n,2}(1)s_np_n\,.
\end{align}
From Theorem~\ref{th:solve} and Lemma~\ref{Lemma01},  we obtain the exact expressions for $p_{n+1,i}(01)$, $p_{n+1,i}(02)$, $p_{n+1,i}(20)$, $l_{n+1,i}(01)$, and thus for $P_{n+1,i}(01)$, $P_{n+1,i}(02)$, $P_{n+1,i}(20)$, $L_{n+1,i}(01)$.
\end{proof}

\subsection{Determination of $S_{n,i}(\bm{\alpha})$ for an arbitrary vertex $\bm{\alpha}$}

{
We finally calculate $S_{n,i}(\bm{\alpha})$ for an arbitrary vertex $\bm{\alpha}$.
Note that a vertex $\bm{\alpha}$ in $H_n$ has a label 
$\gamma_1\gamma_2\cdots\gamma_p$ with length $p$, 
where $1\leq p \leq n$ and $\gamma_z \in \{0,1,2\}$ for $1\leq z \leq p$.  
In the preceding subsections we have determined the degree distribution among all spanning trees for the three outmost vertices corresponding to the case $p=1$ and the six connecting vertices associated with the case $p=2$. 
Next, we will show that for any vertex $\gamma_1\gamma_2\cdots\gamma_p\gamma_{p+1}$ in $H_{n+1}$ with label length $p+1$, 
$S_{n+1,i}(\gamma_1\gamma_2\cdots\gamma_p\gamma_{p+1})$ 
is obtained from some related quantities for the vertex  $\gamma_1\gamma_3\cdots\gamma_p\gamma_{p+1}$ in $H_{n}$ with label length $p$.}

{
Let $\gamma$ be a sequence of $\{0,1,2\}$, and let $|\gamma|$ be the length of $\gamma$ satisfying  $0\leq |\gamma|\leq n-1$. Then, all vertices in $H_n$ have the label form $0 \gamma$, $1\gamma$, or $2\gamma$, while  all vertices in $H_{n+1}$ have the label form  $0k\gamma$, $1k\gamma$, or $2k\gamma$ with $k \in \{0,1,2\}$, corresponding to vertices in $H_{n}^1$, $H_{n}^2$, and $H_{n}^3$ that form $H_{n+1}$. Below we are only concerned with $S_{n+1,i}(0k\gamma)$, since $S_{n+1,i}(1k\gamma)$ and $S_{n+1,i}(2k\gamma)$ can be easily obtained from  $S_{n+1,i}(0k\gamma)$ by symmetry.}

For an arbitrary vertex $\bm{\alpha}$ in $H_n$, we define the following row vector, which contains all quantities we are interested
in:
\begin{align} \label{Matrix0a}
\bm{M}_{n,i}(\bm{\alpha}) =
\left[\begin{matrix}
S_{n,i}(\bm{\alpha}) & P_{n,i}(\bm{\alpha}) & T_{n,i}(\bm{\alpha}) &  R_{n,i}(\bm{\alpha}) & L_{n,i}(\bm{\alpha})\end{matrix}\right],\nonumber
\end{align}
where $n\geq1$. Then, our task is reduced to evaluating
$\bm{M}_{n+1,i}(0k\gamma)$.

Before giving our main result for this subsection, we introduce some matrices. Let $\bm{E}_0$ be the $5 \times 5$ identity matrix, and let $\bm{E}_1$ ($\bm{E}_2$) be  an elementary matrix obtained by interchanging the third (second) column and the fourth column of $\bm{E}_0$. In other words,
\begin{align}
\bm{E}_0=\left[\begin{matrix}1&0&0&0&0\\0&1&0&0&0\\0&0&1&0&0\\0&0&0&1&0\\0&0&0&0&1\end{matrix}\right],\quad \bm{E}_1=\left[\begin{matrix}1&0&0&0&0\\0&1&0&0&0\\0&0&0&1&0\\0&0&1&0&0\\0&0&0&0&1\end{matrix}\right],\quad
\bm{E}_2=\left[\begin{matrix}1&0&0&0&0\\0&0&0&1&0\\0&0&1&0&0\\0&1&0&0&0\\0&0&0&0&1\end{matrix}\right]. \nonumber
\end{align}
Moreover, for any non-negative integer $n$, the matrix $\bm{C}_n$ is defined as
\begin{equation} \label{CN}
\begin{split}
\bm{C}_n=\hspace{13.3cm}
  \\
\left[\begin{matrix}
\frac{2\cdot 5^n+ 3^{n}}{3\cdot5^n} & \frac{3\cdot25^n-9^n}{5^n (5^{n+1}-3^{n+1})} & \frac{3\cdot25^n-9^n}{5^n (5^{n+1}-3^{n+1})} & \frac{(5^n+3^n)^2}{2\cdot5^n (5^{n+1}-3^{n+1})} & \frac{6\cdot25^n-2\cdot9^n}{(5^{n+1}-3^{n+1})^2}\\
\frac{5^n-3^n}{6\cdot5^n} & \frac{9^n-4\cdot15^n+3\cdot25^n}{2\cdot5^n (5^{n+1}-3^{n+1})} & \frac{(5^n-3^n)^2}{2\cdot5^n (5^{n+1}-3^{n+1})} & \frac{25^n-9^n}{2\cdot5^n (5^{n+1}-3^{n+1})} & \frac{7\cdot 125^n+45^n+27^n-3^{n+2}\cdot25^n}{2\cdot5^n (5^{n+1}-3^{n+1})^2}\\
\frac{5^n-3^n}{6\cdot5^n} & \frac{(5^n-3^n)^2}{2\cdot5^n (5^{n+1}-3^{n+1})} & \frac{9^n-4\cdot15^n+3\cdot25^n}{2\cdot5^n (5^{n+1}-3^{n+1})} &  \frac{25^n-9^n}{2\cdot5^n (5^{n+1}-3^{n+1})}  & \frac{7\cdot125^n+45^n+27^n-3^{n+2}\cdot25^n}{2\cdot5^n (5^{n+1}-3^{n+1})^2}\\
0 & 0 & 0 & \frac{2\cdot25^n-9^n-15^n}{5^n (5^{n+1}-3^{n+1})} & \frac{2\cdot(5^n-3^n) (3\cdot25^n-9^n)}{5^n (5^{n+1}-3^{n+1})^2}\\
0 & 0 & 0 &
\frac{3\cdot(5^n-3^n)^2}{2\cdot5^n (5^{n+1}-3^{n+1})} &
\frac{3\cdot(2\cdot5^n-3^n) (5^n-3^n)^2}{5^n (5^{n+1}-3^{n+1})^2}
\end{matrix}\right]. \nonumber
\end{split}
\end{equation}
Then, by alternatively computing $\bm{M}_{n+1,i}(00\gamma)$, $\bm{M}_{n+1,i}(01\gamma)$ and $\bm{M}_{n+1,i}(02\gamma)$, we obtain $\bm{M}_{n+1,i}(0k\gamma)$, as the following lemma states.
\begin{lemma}\label{Theo3}
For the Tower of Hanoi graph $H_n$ and $n > 2$,
\begin{equation}\label{Matrix0b}
\bm{M}_{n+1,i}(0k\gamma)=\bm{M}_{n,i}(0\gamma)\bm{E}_k \bm{C}_n
\end{equation}
holds for all $i=1,2,3$ and $k=0,1,2$.
\end{lemma}
\begin{proof}
We first prove the case $k=0$.

For this case, based on Figs.~\ref{fn},~\ref{pn},~\ref{ln}, and~\ref{lnplus},
we can establish the following relations:
\begin{align*} \label{vectora}
s_{n+1,i}(00\gamma) =  3s_{n,i}(0\gamma)s_n^2+p_{n,i}(0\gamma)s_n^2+t_{n,i}(0\gamma)s_n^2+
4s_{n,i}(0\gamma)s_np_n \,,
\end{align*}
\begin{align}
p_{n+1,i}(00\gamma)&= s_{n,i}(0\gamma)s_n^2+p_{n,i}(0\gamma)s_n^2+6s_{n,i}(0\gamma)s_np_n +3p_{n,i}(0\gamma)s_np_n+ \nonumber \\
&\quad t_{n,i}(0\gamma)s_np_n+3s_{n,i}(0\gamma)p_n^2+ s_{n,i}(0\gamma)s_nl_n \,,\nonumber
\end{align}
\begin{align} 
t_{n+1,i}(00\gamma)&= s_{n,i}(0\gamma)s_n^2+t_{n,i}(0\gamma)s_n^2+6s_{n,i}(0\gamma)s_np_n+p_{n,i}(0\gamma)s_np_n+\nonumber \\
&\quad  3t_{n,i}(0\gamma)s_np_n+3s_{n,i}(0\gamma)p_n^2+s_{n,i}(0\gamma)s_nl_n \,,\nonumber
\end{align}
\begin{align} 
r_{n+1,i}(00\gamma)&= s_{n,i}(0\gamma)s_n^2+p_{n,i}(0\gamma)s_n^2+3r_{n,i}(0\gamma)s_n^2+t_{n,i}(0\gamma)s_n^2+\nonumber \\
&\quad l_{n,i}(0\gamma)s_n^2+2s_{n,i}(0\gamma)s_np_n+p_{n,i}(0\gamma)s_np_n+\nonumber \\
&\quad 4r_{n,i}(0\gamma)s_np_n+t_{n,i}(0\gamma)s_np_n+
 s_{n,i}(0\gamma)p_n^2 \,,\nonumber
\end{align}
\begin{align} 
l_{n+1,i}(00\gamma)&= s_{n,i}(0\gamma)s_n^2+p_{n,i}(0\gamma)s_n^2+2r_{n,i}(0\gamma)s_n^2+t_{n,i}(0\gamma)s_n^2+\nonumber \\
&\quad l_{n,i}(0\gamma)s_n^2+8s_{n,i}(0\gamma)s_np_n+6p_{n,i}(0\gamma)s_np_n+12r_{n,i}(0\gamma)s_np_n+  \nonumber \\
&\quad 6t_{n,i}(0\gamma)s_np_n+4l_{n,i}(0\gamma)s_np_n+12s_{n,i}(0\gamma)p_n^2+4p_{n,i}(0\gamma)p_n^2+\nonumber \\
&\quad 6r_{n,i}(0\gamma)p_n^2+4t_{n,i}(0\gamma)p_n^2+2s_{n,i}(0\gamma)s_nl_n+p_{n,i}(0\gamma)s_nl_n+\nonumber \\
&\quad 2r_{n,i}(0\gamma)s_nl_n+p_{n,i}(0\gamma)s_nl_n+4p_{n,i}(0\gamma)p_nl_n.\nonumber
\end{align}
By definition of $S_{n,i}(\bm{\alpha})$, $P_{n,i}(\bm{\alpha})$, $T_{n,i}(\bm{\alpha})$, $R_{n,i}(\bm{\alpha})$, and $L_{n,i}(\bm{\alpha})$,
we have
\begin{align}
S_{n+1,i}(00\gamma) &=  \frac{2\cdot 5^n+ 3^{n}}{3\cdot5^n} S_{n,i}(0\gamma)+\frac{5^n-3^n}{6\cdot5^n}P_{n,i}(0\gamma)+   \frac{5^n-3^n}{6\cdot5^n}T_{n,i}(0\gamma)\,,\nonumber
\end{align}
\begin{align} 
P_{n+1,i}(00\gamma)&= \frac{3\cdot25^n-9^n}{5^n (5^{n+1}-3^{n+1})}S_{n,i}(0\gamma)+ \frac{9^n-4\cdot15^n+3\cdot25^n}{2\cdot5^n (5^{n+1}-3^{n+1})}P_{n,i}(0\gamma)+\nonumber \\
&\quad \frac{(5^n-3^n)^2}{2\cdot5^n (5^{n+1}-3^{n+1})}T_{n,i}(0\gamma)\,,\nonumber
\end{align}
\begin{align} 
T_{n+1,i}(00\gamma)&= \frac{3\cdot25^n-9^n}{5^n (5^{n+1}-3^{n+1})}S_{n,i}(0\gamma)+ \frac{(5^n-3^n)^2}{2\cdot5^n (5^{n+1}-3^{n+1})}P_{n,i}(0\gamma)+\nonumber \\
&\quad \frac{9^n-4\cdot15^n+3\cdot25^n}{2\cdot5^n (5^{n+1}-3^{n+1})}T_{n,i}(0\gamma)\,,\nonumber
\end{align}
\begin{align} 
R_{n+1,i}(00\gamma)&=  \frac{(5^n-3^n)^2}{2\cdot5^n (5^{n+1}-3^{n+1})}S_{n,i}(0\gamma)+  \frac{25^n-9^n}{2\cdot5^n (5^{n+1}-3^{n+1})}P_{n,i}(0\gamma)+\nonumber \\
&\quad \frac{25^n-9^n}{2\cdot5^n (5^{n+1}-3^{n+1})}T_{n,i}(0\gamma)+ \frac{2\cdot25^n-9^n-15^n}{5^n (5^{n+1}-3^{n+1})}R_{n,i}(0\gamma)+\nonumber \\
&\quad \frac{3\cdot(5^n-3^n)^2}{2\cdot5^n (5^{n+1}-3^{n+1})}L_{n,i}(0\gamma)\,,\nonumber
\end{align}
\begin{align} 
L_{n+1,i}(00\gamma)&= \frac{6\cdot25^n-2\cdot9^n}{(5^{n+1}-3^{n+1})^2}S_{n,i}(0\gamma)+\nonumber \\
&\quad \frac{7\cdot125^n+45^n+27^n-3^{n+2}\cdot25^n}{2\cdot5^n (5^{n+1}-3^{n+1})^2}P_{n,i}(0\gamma)+\nonumber \\
&\quad \frac{7\cdot125^n+45^n+27^n-3^{n+2}\cdot25^n}{2\cdot5^n (5^{n+1}-3^{n+1})^2}T_{n,i}(0\gamma)+\nonumber \\
&\quad \frac{2\cdot(5^n-3^n)\cdot(3\cdot25^n-9^n)}{5^n (5^{n+1}-3^{n+1})^2}R_{n,i}(0\gamma)+\nonumber \\
&\quad \frac{3\cdot(2\cdot5^n-3^n) (5^n-3^n)^2}{5^n (5^{n+1}-3^{n+1})^2}L_{n,i}(0\gamma),\nonumber
\end{align}
which can be rewritten in matrix form as
\begin{equation} \label{SPLM}
\bm{M}_{n+1,i}(00\gamma)=\bm{M}_{n,i}(0\gamma) \bm{C}_n=\bm{M}_{n,i}(0\gamma)\bm{E}_0 \bm{C}_n\,.
\end{equation}
In this way, we have completed the proof of the case $k=0$. 
For the other two cases $k=1$ and $k=2$, the proof is completely analogous to the case $k=0$, 
we omit the details here.
\end{proof}

The first column of the matrix in Eq.~(\ref{Matrix0b}) gives $S_{n+1,i}(0k\gamma)$ for any vertex $0k\gamma$, which is recursive expressed in terms of the related quantities for vertex $0\gamma$. Let $\bm{e}_1$ denote the vector $(1,0,0,0,0)^{\top}$, from  Lemma~\ref{Theo3}, we have the following result.
\begin{theorem}\label{Theo4}
For the Tower of Hanoi graph $H_n$ and $n > 2$,
\begin{equation}\label{Matrixa}
\left[\begin{matrix}S_{n+1,i}(00\gamma)\\S_{n+1,i}(01\gamma)\\S_{n+1,i}(02\gamma)\end{matrix}\right]=
\left[\begin{matrix}\bm{M}_{n,i}(0\gamma)\bm{E}_0\bm{C}_n\\\bm{M}_{n,i}(0\gamma)\bm{E}_1\bm{C}_n\\\bm{M}_{n,i}(0\gamma)\bm{E}_2\bm{C}_n\end{matrix}\right]\times \bm{e}_1
\end{equation}
holds for all $i=1,2,3$.
\end{theorem}
By symmetry, we can obtain the recursive relations for $S_{n+1,i}(1k\gamma)$ and $S_{n+1,i}(2k\gamma)$. Since for arbitrary $n$ and $|\gamma|=0$ and 1, the terms on the right-hand side of Eq.~(\ref{Matrixa}) have been previously determined,
we can repeatedly apply Theorem~\ref{Theo4} to obtain $S_{n,i}(\bm{\alpha})$ for any vertex $\bm{\alpha}$ in $H_n$.

\section{Conclusion}

In this paper we have found the number of spanning trees of the Hanoi
graph by using a direct combinatorial method, based on its self-similar structure, which
allows us to obtain an analytical exact expression for any number of discs.
The knowledge of exact number of spanning trees for the Hanoi graph shows
that their spanning tree entropy is lower than those in other graphs with the same average degree.
Our method could be used to further study in this graph, and other self-similar graphs, their spanning forests, connected
spanning subgraphs, vertex or edges coverings. We have used it to provide a recursive solution for the degree probability distribution
for any vertex on  all spanning tree configurations of the Hanoi graph.

\section*{Acknowledgments}

The authors are grateful to an anonymous reviewer for its valuable comments and suggestions, which have contributed significantly to the readability of this paper. This work was supported by the National Natural Science
Foundation of China under grant No. 11275049. F.C. was supported by the
Ministerio de Economia y Competitividad (MINECO), Spain, and the European
Regional Development Fund under project MTM2011-28800-C02-01.



\end{document}